%% file: attestation.tex
\documentclass[conference]{IEEEtran}
\IEEEoverridecommandlockouts
\usepackage{cite}
\usepackage[latin1,utf8]{inputenc}
\usepackage{amssymb,amsmath,amsthm, amsfonts}
\usepackage{algorithmic}
\usepackage{graphicx}
\usepackage{textcomp}
\usepackage{xcolor}
\def\BibTeX{{\rm B\kern-.05em{\sc i\kern-.025em b}\kern-.08em
		T\kern-.1667em\lower.7ex\hbox{E}\kern-.125emX}}
\usepackage{epsfig,endnotes}
\usepackage{amssymb}
\usepackage{dblfloatfix}
\usepackage{url}
\usepackage{booktabs}
\usepackage{relsize}
\usepackage{multirow}
\usepackage{listings}
\usepackage{siunitx}
\usepackage{url}
\usepackage{comment}
\usepackage[ruled,linesnumbered,lined,boxed]{algorithm2e}
\usepackage{tabularx}
\usepackage[T1]{fontenc}
\usepackage[ngerman,english]{babel}
\usepackage{tikz}
\usepackage{ragged2e}
\usepackage{amsmath}
\usepackage{subcaption}
\usepackage{cleveref}
\crefname{section}{§}{§§}
\newtheorem{theorem}{Theorem}
\begin{document}

\title{slimIoT: Scalable Lightweight Attestation Protocol For the Internet of Things
}

\author{\IEEEauthorblockN{Mahmoud Ammar}
	\IEEEauthorblockA{imec-DistriNet, KU Leuven \\
		mahmoud.ammar@cs.kuleuven.be}
	\and
	\IEEEauthorblockN{Mahdi Washha}
	\IEEEauthorblockA{IRIT, Toulouse University \\
		mahdi.washha@irit.fr}
	\and
	\IEEEauthorblockN{Gowri Sankar Ramachandran}
	\IEEEauthorblockA{University of Southern California \\
		gsramach@usc.edu}\\
	\and
		\IEEEauthorblockN{Bruno Crispo}
	\IEEEauthorblockA{imec-DistriNet, KU Leuven \\
		University of Trento, Italy\\
		bruno.crispo@unitn.it}
	
}

\IEEEoverridecommandlockouts
\IEEEpubid{\makebox[\columnwidth]{978-1-5386-5790-4/18/\$31.00~\copyright2018 IEEE \hfill} \hspace{\columnsep}\makebox[\columnwidth]{ }}	

\maketitle              
\IEEEpubidadjcol
\begin{abstract}
The Internet of Things (IoT) is increasingly intertwined with critical industrial processes, yet contemporary IoT devices offer limited security features, creating a large new attack surface. Remote attestation is a well-known technique to detect cyber threats by remotely verifying the internal state of a networked embedded device through a trusted entity. 
Multi-device attestation has received little attention although current single-device approaches show limited scalability in  IoT applications. Though recent work has yielded some proposals for scalable attestation, several aspects remain unexplored, and thus more research is required. 

This paper presents slimIoT, a scalable lightweight  attestation protocol that is suitable for all IoT devices. slimIoT depends on an efficient broadcast authentication  scheme along with symmetric key cryptography. It is resilient against a strong adversary with physical access to the IoT device. Our protocol is informative in the sense that it identifies the precise status of every device in the network. We implement and evaluate slimIoT considering many factors. On the one hand, our evaluation results show a low overhead in terms of memory footprint and runtime. On the other hand, simulations demonstrate that slimIoT is scalable, robust and highly efficient to be used in static and dynamic networks consisting of thousands of heterogenous IoT devices.

\end{abstract}


\input{sections/introduction.tex}

\input{sections/related.tex}

\input{sections/preliminaries.tex}
\input{sections/slimiot.tex}
\input{sections/implementation.tex}

\input{sections/evaluation.tex}

\input{sections/conclusion.tex}

\bibliographystyle{ieeetr}

\bibliography{references}

\end{document}

%% file: sections/introduction.tex
\section{Introduction}
\label{sec:introduction}
  
The IoT envisions a future where billions of Internet-connected devices are deployed in our environment to support novel cyber-physical applications. Contemporary IoT networks are large and growing in the scale of deployment such as smart buildings and smart cities. IoT applications are naturally distributed and often embedded in numerous heterogeneous computing devices deployed over wide geographical areas. Despite their specialized nature in terms of limited resources and computing power, these devices are becoming attractive targets for a wide variety of cyber attacks \cite{stuxnet} with potentially very dangerous consequences as they process privacy-sensitive information and perform safety-critical tasks that may endanger the lives of many people. 
Therefore, security measures that detect and mitigate cybersecurity threats should take place. 

Remote Attestation (RA) is a common detection and mitigation technique for exposing the misbehaviour of a compromised IoT device, where a trusted entity, denoted as a \textit{verifier}, checks the integrity of the internal state of an untrusted remote device, denoted as a \textit{prover}. 
Most research to date has focused on attesting a single $prover$ device~\cite{h2010bootstrapping,spab,smart,smarm}, whereas relatively little attention has been paid to the scalability issue, even though real IoT applications deploy devices in massive numbers forming large mesh networks or swarms where physical access to any device is easily reachable.


The recently proposed attestation schemes \cite{seda2015,sana2016,darpa2016,lisa2017,scapi2017,salad,wise} have addressed the attestation  problem at a large scale where a significant number of devices have to be attested efficiently and securely. 
To narrow the focus, only three approaches, namely DARPA \cite{darpa2016}, SCAPI \cite{scapi2017}, and SALAD \cite{salad}, have considered both remote and physical attacks. The latter is expensive in terms of using public key cryptography which is not suitable for low-end embedded devices. SCAPI is a scalable attestation protocol that handled the shortcomings of DARPA. In spite of the strong security offered, SCAPI still depends on some hard assumptions (e.g. requires half of the devices in a network to be uncompromised) and incurs high overhead in terms of memory footprint and power consumption due to neighboring discovery and exchanging many channel keys. 

In this paper, we present \textit{slimIoT}, a scalable lightweight attestation protocol that can be implemented efficiently also on common resource-constrained IoT devices (e.g. IETF Class-1 \cite{RFC7228}) with strong security guarantees. slimIoT classifies the connected devices in the swarm into clusters where all clusters are periodically attested against physical attacks under the assumption that a physical attack is a time-consuming activity \cite{abdec2008,darpa2016}, and one or more clusters are attested against remote attacks every attestation period. 
Our protocol relies on authenticated parameterized broadcast messages using symmetric key cryptography. Authenticated broadcast techniques require an asymmetric encryption to provide strong security guarantees. Otherwise, any compromised device can render the entire swarm insecure. We achieve asymmetry through a delayed disclosure of symmetric keys generated by using a one-way hash-chain mechanism at the \textit{verifier} side. 
In line with a similar technique for authentication \cite{spins}, slimIoT requires that all devices are loosely time synchronized. To detect physical attacks, 
we take advantages of the $nonce$ value used to avoid replay attacks. All $nonces$ in various attestation periods are updated securely in a linked way, where missing one update prevents the corresponding $prover$ from teaming up with the swarm anymore as it is most likely physically-compromised due to its absence for a reasonable amount of time. 

We show that slimIoT is secure, scalable, robust, and runs efficiently on mesh-networked low-end embedded devices. Scalability and robustness are demonstrated by simulating large mesh networks using OMNeT++ framework \cite{omnet}.


\textbf{Paper outline.} The remainder of this paper is organized as follows. Section~\ref{sec:relatedwork} reviews the related work. Preliminaries are presented in Section~\ref{sec:preliminaries}. Section~\ref{sec:slimiot} describes slimIoT in details.  Implementation details and evaluations are reported in Section \ref{sec:implementation} and \ref{sec:evaluation} respectively. Section~\ref{sec:conclusions} concludes.

%% file: sections/related.tex
\section{Related Work}
\label{sec:relatedwork}

Prior work in single-prover RA can be divided into three approaches: hardware-based~\cite{h2010bootstrapping}, software-based~\cite{spab}, and hybrid~\cite{smart,smarm}. Each of these approaches has pros and cons. Nevertheless, all of them target a single-device attestation. Accordingly, they are not efficient for collective attestation. 


Recently, few research papers have been published to address the problem of swarm attestation. Depending on public key cryptography, SEDA~\cite{seda2015} lets each device in the network attests its neighbours and propagates the aggregated attestation reports back to its parent till eventually received by the \textit{verifier}. SANA~\cite{sana2016} enhanced SEDA by providing an efficient and scalable attestation technique based on a novel signature scheme that enables anyone to publicly verify the attestation reports in a very efficient way. The authors in~\cite{lisa2017}  have presented two different swarm attestation protocols, called \textit{{LISA}$\alpha$} and \textit{{LISA}s} using symmetric key cryptography and depending on SMART architecture~\cite{smart}. WISE \cite{wise} is the first smart swarm attestation protocol that minimizes the communication overhead by attesting some devices only (instead of the entire swarm) while preserving an adjustable level of security.

All the aforementioned swarm attestation proposals have assumed software-only adversary model. In mesh networks this assumption may not hold as it is easy for an adversary to capture a device and physically tamper with it. SALAD \cite{salad} targets highly dynamic networks in a distributed manner depending on public key cryptography, with limited scalability in networks containing IoT devices with small memory footprint to exchange and store a swarm-wide attestation report.

DARPA~\cite{darpa2016} mitigates invasive physical attacks by combining the current scalable attestation protocols, e.g. SEDA~\cite{seda2015}, with an absence detection protocol ~\cite{abdec2008} under the  assumption that the physical attack consumes an amount of time, whose lower bound is known a priori, during which, the compromised prover device is offline. 

SCAPI~\cite{scapi2017} has enhanced DARPA by proposing a better way to detect physical attacks, where all devices in the swarm share two common session keys that are periodically updated by a leader device. Building on the assumption that the physical tampering with the IoT device requires turning it offline for a considerable amount of time, SCAPI detects physically-compromised devices by periodically updating the group-wide secrets in a way that all offline devices cannot get this update and as a consequence they cannot team up with the swarm once they are online again. Both session-key update and attestation phases in SCAPI rely on mutually authenticating neighboring devices by exchanging channel keys used for encrypting all messages exchanged.
In nutshell, we show that slimIoT overcomes the limitations of SCAPI as follows:
\begin{itemize}
	
	\item It precisely attests the devices in the network and identifies the ones that run a compromised software or have been physically manipulated, assuming that there is only one device within the range of the $verifier$ uncompromised. This relaxes the assumption of SCAPI that requires at least half of the devices to be healthy.
	\item It is more efficient for highly dynamic and heterogenous networks where devices do not have to know about their neighbours or exchange any channel key. Accordingly, slimIoT requires less memory footprint where the space needed to store cryptographic keys is constant in all devices in the swarm. 
	\item It is more robust to false positives i.e., healthy devices that are regarded as physically-compromised, as we explain later on. 
	   
\end{itemize} 



%% file: sections/preliminaries.tex
\section{Preliminaries}
\label{sec:preliminaries}

\textbf{Network and Device Requirements.} We consider that  devices in a swarm $S$ can have various software configurations and possess heterogenous hardware capabilities. However, we assume that each device $D_i$ in $S$ satisfies the minimum hardware properties required for secure RA~\cite{min2014} which are ROM and memory protection unit (MPU). The ROM is needed to store the attestation routine and cryptographic keys, whereas the MPU is required to enforce access control over secret data. The swarm topology can be either static or dynamic. However, we assume that devices stay connected to the network during motion and there is at least one device that is physically uncompromised and reachable by the $verifier$. Furthermore, we assume that all devices are loosely timed synchronized with the $verifier$. 
Henceforth, we refer to the $verifier$ as $\upsilon$ and the $prover$ ($D_i$) as $\rho$.
\newline
\textbf{Adversary Model.} We consider a strong adversary $adv$, who has the ability to perform remote and invasive physical attacks. We assume that the remote $adv$ has full access to the network and can either perform passive (e.g. eavesdrop on communication, etc.) or active (e.g. inject a malware, etc.) attacks. Additionally, the physical attacker is able to physically compromise any device $D_i$ (except one) in $S$, turn it off, and perform any operation over it (e.g. learn all stored secrets). Similar to all prior work \cite{darpa2016,scapi2017}, we build on the assumption that a physical attack requires at least an amount of time that is known a priori, denoted as  $T_{adv}$. We rule out DoS and non-invasive physical attacks (e.g. side channels). 
\newline
\textbf{Attestation Considerations.} We assume that $\upsilon$  is a powerful device (e.g. Raspberry pi or higher) and is unaware of the current network topology of the swarm due to it's dynamicity. If the IoT device lacks a Trusted Execution Environment (TEE) (e.g. Class-1 devices \cite{RFC7228}), we assume that the attestation code should be correct, atomic, leak no sensitive information, and leave no traces in the RAM after execution (for example, by employing the security MicroVisor \cite{smv}). 

%% file: sections/slimiot.tex
\section{slimIoT: Protocol Description}
\label{sec:slimiot}

\textbf{Overview.} slimIoT is a scalable attestation protocol that aims to efficiently attest multiple devices connected in a mesh network and detect the compromised ones. It consists of two different phases. The first phase (\cref{init}) takes place once, before the deployment, where all devices are initialized by a trusted party, $\upsilon$, with some public write-protected and private data. Upon deployment, all devices in the swarm are loosely-time synchronized with $\upsilon$ using a write-protected real time clock. In the second phase (\cref{attestphase}), $\upsilon$ periodically performs the attestation routine  in order to learn the precise status of every $\rho$ in $S$. The attestation phase is very flexible and can be performed in three ways: (i) identifying only physically-attacked devices through absence detection (physical attestation), (ii) attesting all devices against physical attacks and verifying the software integrity of some devices (clusters) against remote attacks (partial attestation), or (iii) attesting the entire swarm against both physical and remote attacks (full attestation). 

\subsection{Initialization Phase}
\label{init}
\subsubsection{Verifier Setup}
First, $\upsilon$ produces a sequence of secret keys  (one-way keychain) of length $j$ by choosing the last $K_j$ randomly, and generating the remaining values by successively applying a one-way hash function $F$ (e.g. SHA-256) where $K_{j-1}$ = $F(K_j)$. During the deployment, devices are initialized with $K_0$ as a commitment secret value. Using it, they can authenticate all other keys in the keychain by recursively executing $F$ (e.g. $K_0$ = $F(F(F(K_3)))$). However, they cannot compute any of these keys due to the one-way property of $F$.  The remaining keys are used as session keys in sequence, starting with $K_1$ and ending up with $K_j$, to authenticate packets exchanged in the attestation phase. To authenticate successive session keys quickly, the last authenticated session key is always saved at $\rho$ side (it does not replace $K_0$ as it is not secret and always exchanged in plaintext). Whenever $\upsilon$ needs to update the keychain, $K_0$ can be easily updated to the correct value by broadcasting a message containing the new value of $K_0$ and encrypted using the old one. 

Second,  $\upsilon$ divides the time into discrete time intervals, called \textit{epochs}, where the maximum duration of one \textit{epoch}, $T_{att}$, should not last longer than the overall physical attack time, $T_{adv}$. The value of both $T_{att}$ and $T_{adv}$ is application-dependent. The attestation phase is performed periodically at the beginning of every \textit{epoch}. Each \textit{epoch} is divided into multiple unequal discrete time sub-intervals. For the sake of clarity, in our protocol, we assume 4 non-overlapping time sub-intervals per \textit{epoch}. $\upsilon$ associates each session key of the one-way keychain with one time sub-interval as a MAC key to authenticate all packets sent in that period. Upon the expiration of the corresponding time sub-interval, $\upsilon$  exposes the associated session key after a delay $d$. The disclosure delay, $d$, is either a secret variable value that can be adjusted in every new attestation request or a fixed value  tied to the time intervals and broadcasted periodically in a special packet.  In the explanation later and in Figure 1, we deliberately abstract from any implementation details related to  time-related functions due to the limited space.
\subsubsection{Prover Initialization} 
Each $\rho$ is initialized with some public and private values. The public values are (i) its Id ($D_i$), (ii) the Id of the parent node in $S$ (ParID), and (iii) the Id of the cluster (ClusID), where $\rho$ belongs to, in $S$. The value of ParID is initially null and is updated only by the attestation code after forming the spanning tree during attestation. Devices are classified by $\upsilon$ into different clusters where each cluster has a ClusID. This classification is protocol-independent and can be based on various factors (e.g. geographical locations in static networks, common tasks among devices, network traffic distribution, equal clusters, and etc.). 

On the other hand, each $\rho$ is initialized with (i) two device-dependent secret keys for authentication ($K_a$) and software integrity ($K_t$) purposes, (ii) one cluster-dependent secret value ($K_c$) for multicasting and attestation purposes, and (iii) one group-wide secret key ($K_0$), which is the last key in the keychain and used for authenticating the session keys exchanged at the swarm level. Furthermore, every $\rho$ holds a group-wide secret value, $nonce$, that serves as a refreshment value to avoid replay attacks. This value is securely updated  twice in every attestation period. The updates are linked in a way to detect physical attacks, where missing one update prevents the corresponding device from authenticating future messages. Also, $\rho$ stores securely the correct MAC digest of the safe state of its memory, computed using $K_t$, and denoted as $H_S$. 

\begin{figure*}[t!]
	\centering
	\includegraphics[width= \textwidth, height = .7\textwidth]{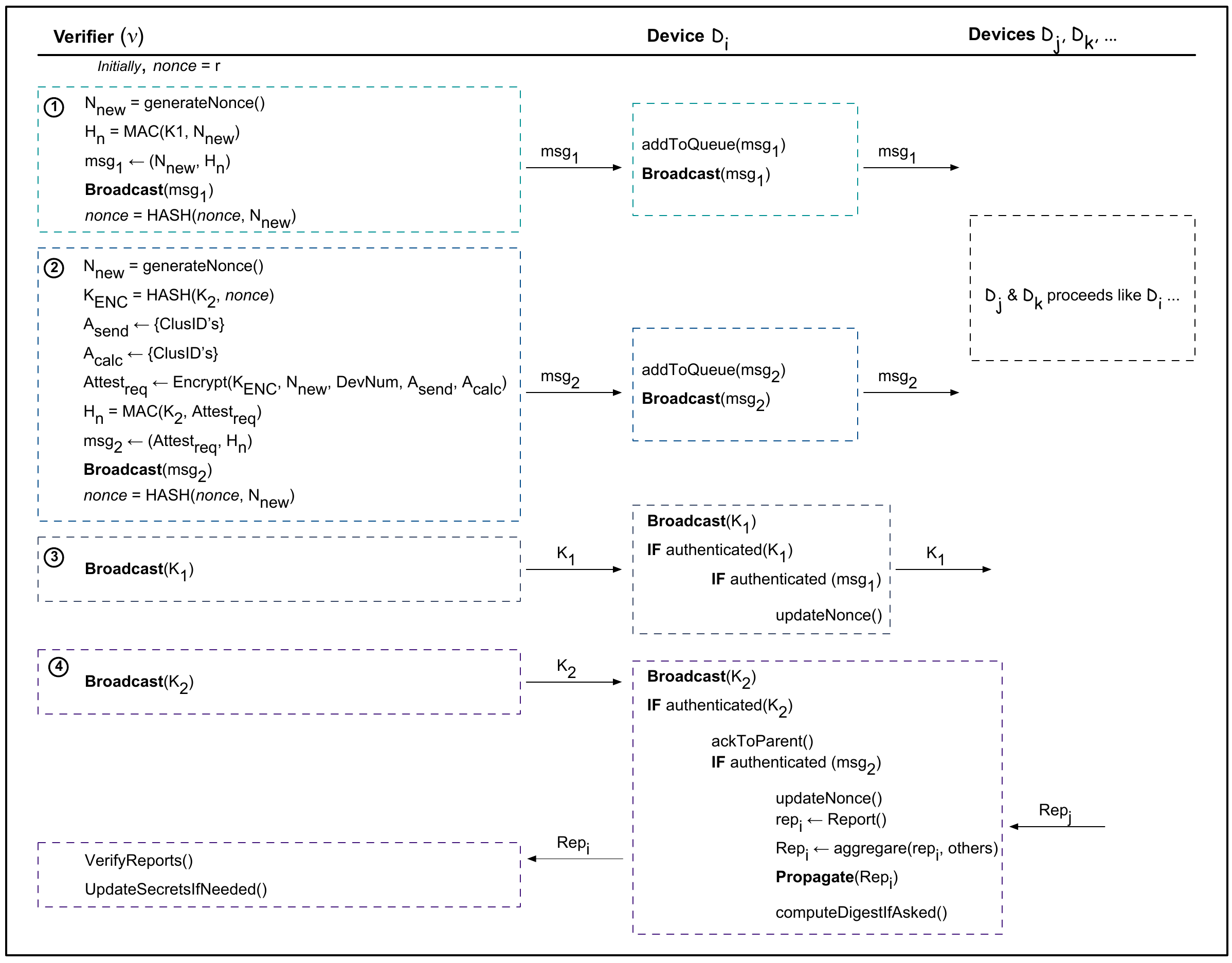}
	\caption{An overview of the working mechanism of slimIoT}
	\label{fig:figure1}
\end{figure*}

\subsection{Attestation Phase}
\label{attestphase}

\textbf{Initiation of Remote Attestation.}  
At the beginning of a new \textit{epoch}  (see \textcircled{1} in Figure \ref{fig:figure1}), $\upsilon$ generates a new random value, $Nnew$,  and computes the MAC of this value with a key that is secret at that point in time, where this key belongs to the keychain generated at the initialization phase. Considering that $\upsilon$ and $\rho$'s are loosely time synchronized, and each node knows the upper bound of the maximum synchronization error, $\upsilon$ broadcasts the new value along with its MAC to all devices within range without revealing the secret key. $\upsilon$ then updates the value of the secret $nonce$ by computing the hash value of the current one concatenated with the random generated value (e.g. $nonce = HASH(nonce||Nnew$). All receiving nodes accept this packet if it complies with their time schedule, store it in a buffer for authentication later on, and re-broadcast it again to their neighboring nodes. 

Upon the start of the next time sub-interval in the current \textit{epoch} (e.g. a few milliseconds later in our implementation) (see \textcircled{2} in Figure \ref{fig:figure1}), $\upsilon$ creates an attestation request ($Attest_{req}$) consisting of (i) a new random value ($Nnew$) used to update the $nonce$ one more time, (ii) the number of devices in $S$ ($DevNum$) used to notify $\rho$'s about the length of the $n-$bit vector that has to be created for acknowledging the physical presence where 1 at position $i$ indicates that $D_i$ is physically present in $S$ and uncompromised, and (iii) two lists of different cluster ID's where $\rho$'s that belong to any of the clusters in the first list ($A_{send}$) have to send an attestation response reflecting the state of their software too, whereas $\rho$'s  that belong to clusters mentioned in the second list ($A_{calc}$) have to compute the MAC digest of the existing software in their memories using $K_t$ as a last step of this attestation routine in order to produce the attestation report quickly in the next attestation period when asked. The calculation of MAC digest comes after propagating the aggregated attestation reports and thus the time consumed is not counted in the total time consumed in attesting the entire swarm (e.g. see \textcircled{4} in Figure \ref{fig:figure1}). In any attestation period, either of both lists can be empty if $\upsilon$ is interested only in attesting the physical presence of devices in $S$. 

$Attest_{req}$ is encrypted using a secret key that is derived from computing the hash value of the next session key in the keychain along with the previously updated $nonce$ ($K_{ENC}$ = $HASH(K_2, nonce)$). $\upsilon$ then computes the MAC of the encrypted $Attest_{req}$ using the same session key ($K_2$) and broadcasts it to all nearby devices. At $\rho$ side, every node stores this request in a buffer if it complies with its time schedule, and re-broadcasts it again in order to be received by all nodes in $S$. As a last step in the current time sub-interval, $\upsilon$ updates the $nonce$ for a second time. Please note that even after revealing $K_2$, $\rho$'s will not be able to derive $K_{ENC}$ if they did not manage to authenticate the previous $Nonce_{Update}$ request and update the $nonce$ successfully.
\newline
\textbf{Key Disclosure.} At the time of key disclosure, $\upsilon$ broadcasts the verification key to all $\rho$'s in vicinity. If this key is received according to the time schedule, the receiving $\rho$ re-broadcasts it to other neighbors. If the authentication process of this key is passed successfully, $\rho$ uses it to authenticate the corresponding packet in the buffer. Upon the successful authentication of any packet, $\rho$ executes the request of that packet. For example, in the third time sub-interval (See \textcircled{3} in Figure \ref{fig:figure1}), $\rho$ performs the $Nonce_{Update}$ request. Using the true updated $nonce$ along with the exposed session key in the fourth interval, $\rho$ is able to derive the decryption key, $K_{ENC}$, and decrypt $Attest_{req}$. 
Every $\rho$ acknowledges about receiving the second session key to the sender device (the parent) and thus the spanning tree is formed and rooted by $\upsilon$. 
\newline
\textbf{Attestation Report.} 
An $n$-bit vector is created by the leaf nodes in the spanning tree where $n$ equals to the number of devices in $S$ mentioned in $Attest_{req}$. Considering that the largest $prover$ ID does not exceed $DevNum$ in $S$, every active $D_i$, during the ongoing attestation, has to set the $i^{\text{th}}$ position in the $n$-bit vector to 1 to avoid being labeled as physically-compromised. 
If any $\rho$ received more than one propagated $n$-bit vector, it aggregates all of them by performing an OR operation.  If any $\rho$ is asked to provide an attestation report verifying its software integrity (by having its cluster ID in the $A_{send}$ list),  it checks the computed MAC digest ($H'_S$) upon the exit of the previous attestation period and compares it with the stored one ($H_S$). If $H'_S$ equals to $H_S$, $\rho$ creates an attestation report that constitutes of its $ID$ and an attest value that equals to the hash value of the computed MAC digest along with the last updated $nonce$ (e.g. attest = $HASH(H'_S || nonce)$). After setting the $i^{\text{th}}$ position in the $n$-bit vector to 1, $\rho$ propagates both the (aggregated) attestation report(s) and the $n$-bit vector to the parent device.  A secure aggregate of multiple attestation reports is computed by XOR-ing all attest values \cite{mac2008} in these reports and maintaining the participating devices IDs in an accompanying list. 
If  $H'_S$ and $H_S$ are not equal, $\rho$ does not create an attestation report of its status as it is most likely remotely-compromised but it still confirms its presence physically by assigning a one value to the $i^{\text{th}}$ position in the $n$-bit vector. 
As a last step after propagating the aggregate of the attestation reports, 
every $\rho$ checks if its cluster ID is included in $A_{calc}$ in $Attest_{req}$. If so, it calculates the MAC digest of its memory using $K_t$, and stores it in an \textit{unprotected} \footnote{It reflects the existence of a malware if it is altered before being used to create the attestation report in the next attestation iteration.} memory area. Having this value accelerates the creation of the attestation report in the next period, thus reducing the runtime overhead of slimIoT as computing MAC is time-consuming activity.  
All communications between $\rho$'s are authenticated using $HASH(K_0 || nonce)$. 
\newline
\textbf{Physical and Remote Attacks Detection.} 
Remotely-compromised devices are detected by being unable to participate in the aggregate of the attestation reports (their ID's are not present!), whereas physically-compromised devices are detected through their absence during the execution of the protocol (zero values present at positions with indexes equal to their ID's in the $n$-bit vector). Upon detection, security measures can take place to recover compromised devices (e.g. secure erasure \cite {speed}). 
\newline
In case of detecting a physical attack, there is a probability that the attacker can decrypt all group-wide secret communications if she was eavesdropping on the communication during  the execution of the attestation process, as explained in \cref{sa}. Therefore, to avoid this vulnerability, $\upsilon$ updates all group-wide secrets (e.g. $nonce$ and $K_0$) by broadcasting a number of messages equal to the number of healthy clusters. Each message contains the new values of secrets and is encrypted using one of the healthy cluster-wide secrets, $K_c$. For operative healthy devices in the compromised cluster, $\upsilon$ creates a message containing new cluster-wide secret value ($K_c$),  encrypts this message with each of the individual healthy devices keys, $K_a$,  appends all these ciphertexts to each other in one big message, and then broadcasts it. After that, $\upsilon$ broadcasts another message containing the new group-wide secrets and encrypted using the updated cluster wide secret key $K_c$. 

\subsection{Security Analysis}
\label{sa}

The goal of a secure attestation scheme is to distinguish between healthy and compromised devices in a swarm {$S$}, where limited false positive cases are accepted but not vice versa. This is formalized by the following adversarial experiment $ATT_{adv}^{n, c}(j)$, where $adv$ interacts with $n$ devices in addition to $\upsilon$, and compromises up to $c$ devices in $S$, 
where $c$ $\leq$ $n -1$. Considering that $adv$ is computationally bounded to the capabilities of the devices deployed in $S$, $adv$ interacts with the devices a polynomial  number of times $j$, where $j$ is a security parameter.  $\upsilon$ verifies the attestation aggregates received from attesting $S$ and outputs a decision as 0 or 1. Denoting the decision made  as $A$, $A$ = 1 means that the attestation routine is finished successfully and all devices are labeled correctly, or $A$ = 0 otherwise. 
According the definition of secure swarm attestation given by \cite{seda2015}, we summarize the security of slimIoT with an informal proof sketch (due to the limited space).

\begin{theorem}[Security of slimIoT]
	slimIoT is a secure scalable attestation protocol if $Pr$[$A$ = 1 | $ATT_{adv}^{n, c}(j)$ == $A$] is negligible for 0 $<$ $c$ $<$ $n$, under the following conditions:
	\begin{itemize}
		\item $T_{att}$ $\leq$ $T_{adv}$. 
		\item The PRNG, MAC, and aggregation schemes used are secure and  selective forgery resistant.
	\end{itemize}
\end{theorem}
\begin{proof}
	
	$adv$ can be either: (i) remote, (ii) physical, or (iii) sophisticated (performs remote and physical attacks simultaneously).
	In the case of a remote $adv$, compromised devices are always detected with a probability 100\% as the attestation routine is atomic and executes in TEE. $adv$ can not bypass the rules enforced by MPU and thus can not learn the stored secrets or break the selective forgery MAC scheme. 
	
	
	Physical attack is a time-consuming activity which requires at least few minutes to retrieve the secrets. Attesting $S$ is guaranteed to be achieved in less time, $T_{att}$. Accordingly, the exposed secrets will be useless for $adv$ as the attestation is performed and the secret $nonce$ is updated. Thus, $adv$ can not decrypt or authenticate any further communications and the compromised device will be detected due to its absence. 
	
	
	In the case of a sophisticated $adv$, she can eavesdrop on the communication and record all exchanged packets while performing a physical attack on one of the devices. The compromised device will be detected by $\upsilon$ as explained above but with the information gained from both attacks, $adv$ is able to synchronize and learn the current secret keys and thus decrypting all future communications and threatening the security of slimIoT if no action is taken. 
	Therefore, in case of detecting a physical attack, the last step of the protocol is securely updating all group-wide secrets in all healthy devices. This means that the probability of $adv$ to get benefits from the exposed group-wide secrets is negligible. 
	
\end{proof}

%% file: sections/implementation.tex
\section{Implementation}
\label{sec:implementation}

We implemented a practical scenario on a 5-node heterogeneous  mesh network where each node belongs to either the Arduino or the MicroPnP IoT platform \cite{mpnp}.  
The Arduino platform offers an 8-bit AVR ATmega 328p microcontroller running at 16 MHz with 2 KB of SRAM and 32 KB of flash memory, whereas the MicroPnP IoT platform provides an 8-bit AVR ATmega 1284p microcontroller running at 10 MHz with 16 KB of SRAM and 128 KB of flash. Both platforms are equipped with IEEE 802.15.4e Time-Slotted Channel Hopping (TSCH) \cite{802.15.4e} radio transceiver for wireless communication.

We used SHA-256 as a one-way hash function for generating keychains, and as a keyed-hash message authentication code (HMAC) to measure the software integrity during the attestation phase. We employed AES-128 in counter mode (CTR) as an authentication encryption scheme.  



%% file: sections/evaluation.tex
\section{Evaluation}
\label{sec:evaluation}
In this section, we first present an analytical model for cluster selection where the communication cost depends on the tolerated latency of detecting remotely-compromised devices. Then, we illustrate our experimental results in terms of performance, scalability and robustness.  
\subsection{Communication Costs versus Latency of Detecting Attacks}
\label{subsec:model}

Let $G=(N,L)$ represents a graph for a mesh network where $N$ is the number of IoT devices (vertices) and $L$ is any number of links (edges) that maintains the connectivity between devices. Considering that the network is divided into $M$ non-overlapping clusters, let $C_{i}=\{n_1,n_2,...\} \subseteq N$ be a set of devices representing the $i^{th}$ cluster such that $\bigcap_{j=1}^M C_{j} = \emptyset $. We model the attestation response at iteration $t$ of the $i^{th}$ cluster selected for attestation by 3-tuple $A_{t}^i=\langle C_{i}, S_{t}^i, T_{las} \rangle $, where $S_{t}^i \in \{0,1\}$ as $0$ value indicates that the corresponding cluster has at least one compromised device in this iteration, whereas $1$ value reflects a safe status. $T_{las}$ is the last time when the cluster has been attested (represented in minutes).

The cluster selection process aims to minimize the communication overhead by attesting  various parts of the network at different time periods while maintaining the security level. However, excluding some clusters from attestation  at iteration $t$ (in spite of its shortness) may increase the latency of identifying the remotely-compromised IoT devices in these clusters. Thus, we deal with the communication overhead issue as a trade-off between the cluster selection and the latency of identifying compromised devices, where we leverage the history of every cluster (i.e., $A_{t-1}^i, A_{t-2}^i$, etc. ). More formally, we handle the selection of clusters as a constrained minimization optimization problem objectively defined as: \vspace{-8pt} 

\begin{equation}
\resizebox{1\hsize}{!}{$
\begin{aligned}
& \underset{\alpha_1,\ldots,\alpha_M}{\text{minimize}}\ \sum_{j=1}^{M}  \alpha_j*P(S_{t}^j=1)  
\ \ \text{s.t. }\
\frac{\sum_{i=1}^{M} \alpha_{i} * |C_{i}|  }{|N|} \geq Tr_{cov},\\
& Max\{ (1-\alpha_i)*(T_{next}-A_{t-1}^i.T_{las}) :\ i\in\{1,\dots,M\}\} \leq T_{Max}
\end{aligned} $}
\end{equation} 
where $\alpha_{\bullet}\in \{0,1\}$ is a free variable corresponding to a cluster, $P(S_{t}^j=1)$ represents the probability of the $j^{th}$ cluster being successfully attested, $Tr_{cov}\in [0,1] $ is a threshold specifying the ratio of devices that will be attested, $T_{next}$ is the next attestation time in minutes, and $T_{Max}\in \mathbb{Z}^+$ is the maximum number of minutes that each cluster can be excluded from the attestation process. The first constraint ($Tr_{cov}$) in the objective function is application dependent and can be tied to the time period of attesting against physical attacks, whereas the second constraint ($T_{Max}$) is added to ensure that all clusters have been attested in a pre-defined maximum time window. In the computation of the probability component stated in the objective function, we first assume that the current state of the cluster is independent from the previous states (i.e. $S_{t-1}^j$). Thus,  we adopt the Maximum Likelihood Estimation (MLE) method \cite{mle} through leveraging the given information history ($\{ A_{1}^j,\dots, A_{t-1}^j \}$) about the state of the cluster at various attestation iterations. More precisely, we first count the number of times that the cluster has been successfully attested as reported in the cluster's history. Then, the resulted number is divided by the number of times that the cluster has been considered in the attestation process. Formally, the probability component for the $j^{th}$ cluster is computed as follows: 
\begin{equation}
P(S_{t}^j=1)=\frac{| \{  A_{i}^j: i\in\{1,...,t-1\},  S_{i}^j=1\}|}{| \{  A_{i}^j: i\in\{1,...,t-1\} \}|}
\end{equation}
where $|\bullet|$ represents the set cardinality (length).  

Selecting values of both constraints ($T_{Max}$ and $Tr_{cov}$) is application-dependent. However, the experimental results in the following section are based on  setting the value of $T_{Max}$  to 60 minutes and the value of $Tr_{cov}$  to be either 0.25 or 1, where 1 shows the worst case in the communication overhead in terms of runtime as all clusters are selected in the attestation. 

\begin{table}[b]
	\caption{Cryptographic Runtime Measurements}	
	\normalsize
	\newcommand{\hd}[1]{\multicolumn{1}{c}{\textbf{#1}}}
	\newcommand{\hdd}[1]{\multicolumn{1}{c}{\multirow{2}{*}{\textbf{#1}}}}
	\resizebox{\columnwidth}{!}{%
		\begin{tabular}{lrrr}
			\toprule
			\hdd{Runtime Measurements} & \hd{Arduino} & \hd{MicroPnP}         \\
			& \hd{ATmega328P} &\hd{ATmega1284P}\\
			\midrule
			Session Key Authentication            & 3.213 ms       & 5.145 ms  \\
			Updating $Nonce$ value            & 6.34 ms       & 10.10 ms  \\
			Deriving $K_{ENC}$, authenticating, and decrypting AttestRequest (64 Bytes)          & 47.38 ms       & 75.9 ms  \\
			preparing and aggregating two attestation reports            & 3.61 ms       & 5.184 ms  \\ Oring two vectors of length 255 Bytes (2000 devices)		& 0.449 ms		& 0.648 ms \\
			Verification of 32 Bytes MAC (for 64 Bytes message)		& 12.7 ms		& 20.36 ms \\
			
			\bottomrule \\
		\end{tabular}
	}
	
	\label{tab:crm}
	\vspace{-20pt}
\end{table}

\subsection{Experimental Results}

\textbf{Runtime Overhead.}
Computing the HMAC-SHA2 digest of MCU with 128 KB flash memory in MicroPnP platform consumed 9.68 seconds, whereas it demanded 1.47 seconds for MCU with 32 KB flash memory in Arduino. Table \ref{tab:crm} presents  the precise time consumed by the various cryptographic operations in slimIoT for both Arduino and MicroPnP platforms. 

\textbf{Network Runtime Measurements.} We measured an average propagation delay between any two neighboring nodes of 17 ms. Moreover, the maximum throughput we measured at the application layer is 56 kbps. However, the maximum throughput of IEEE 802.15.4e in theory is 250 kbps. 

\begin{figure*}[t!]
	\centering
	\begin{subfigure}[b]{0.5\textwidth}
		
		\includegraphics[width= \textwidth]{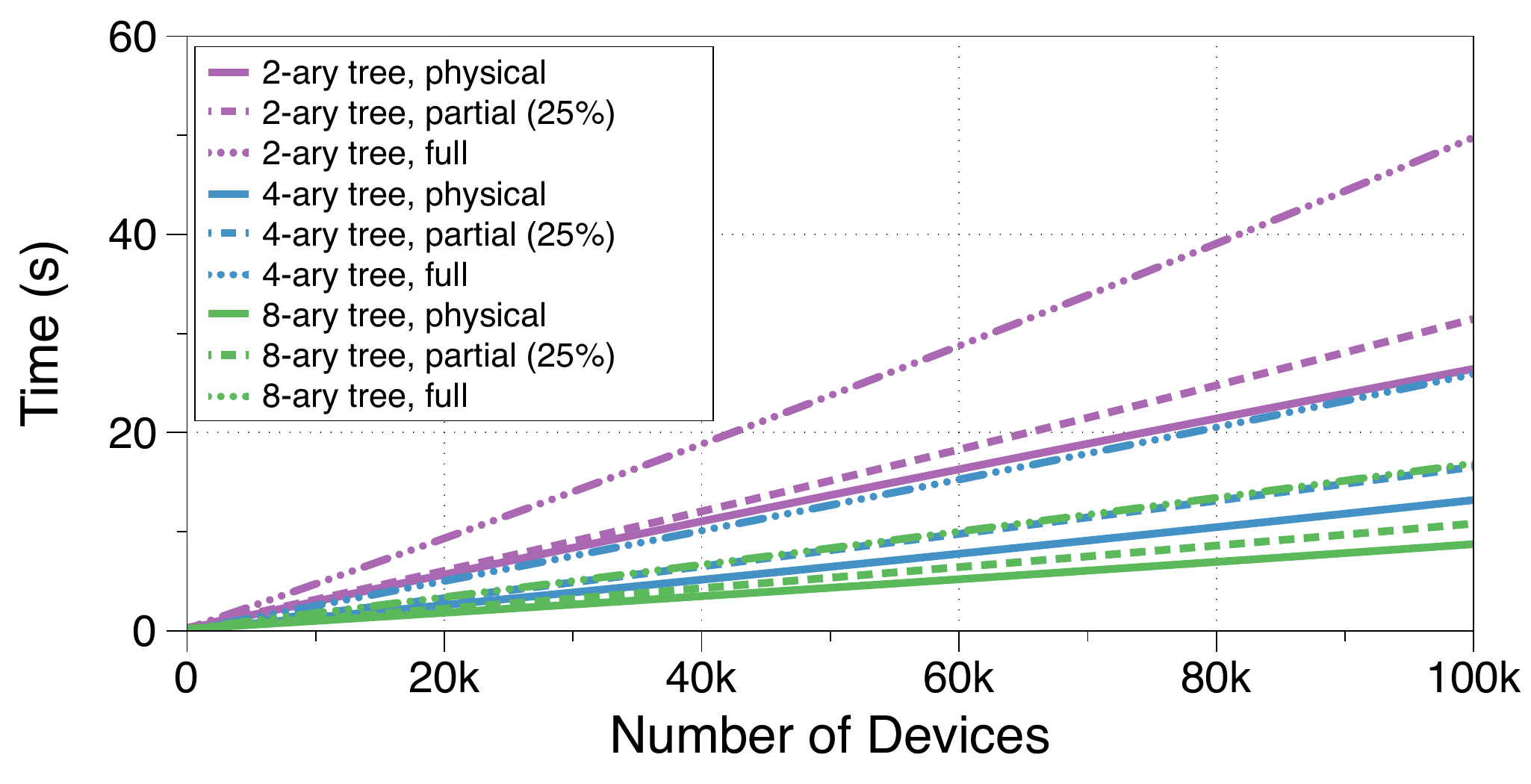}
		\caption{Runtime of various static topologies}
		\label{fig:trees}
	\end{subfigure}
	\begin{subfigure}[b]{0.49\textwidth}
		
		\includegraphics[width= \textwidth]{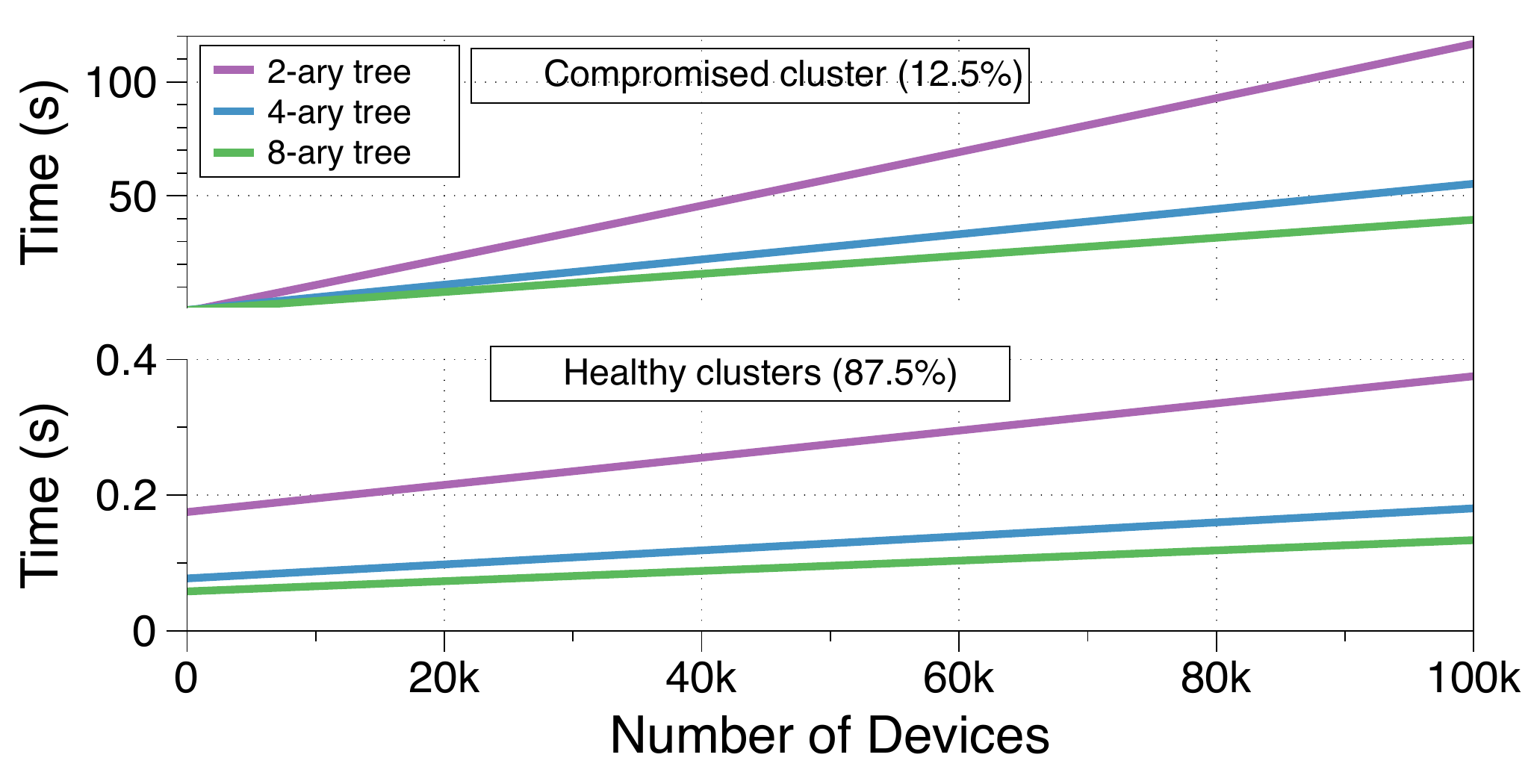}
		\caption{Runtime of updating secrets} 
		\label{fig:updates}
	\end{subfigure}
	\begin{subfigure}[b]{0.5\textwidth}
		
		\includegraphics[width= \textwidth]{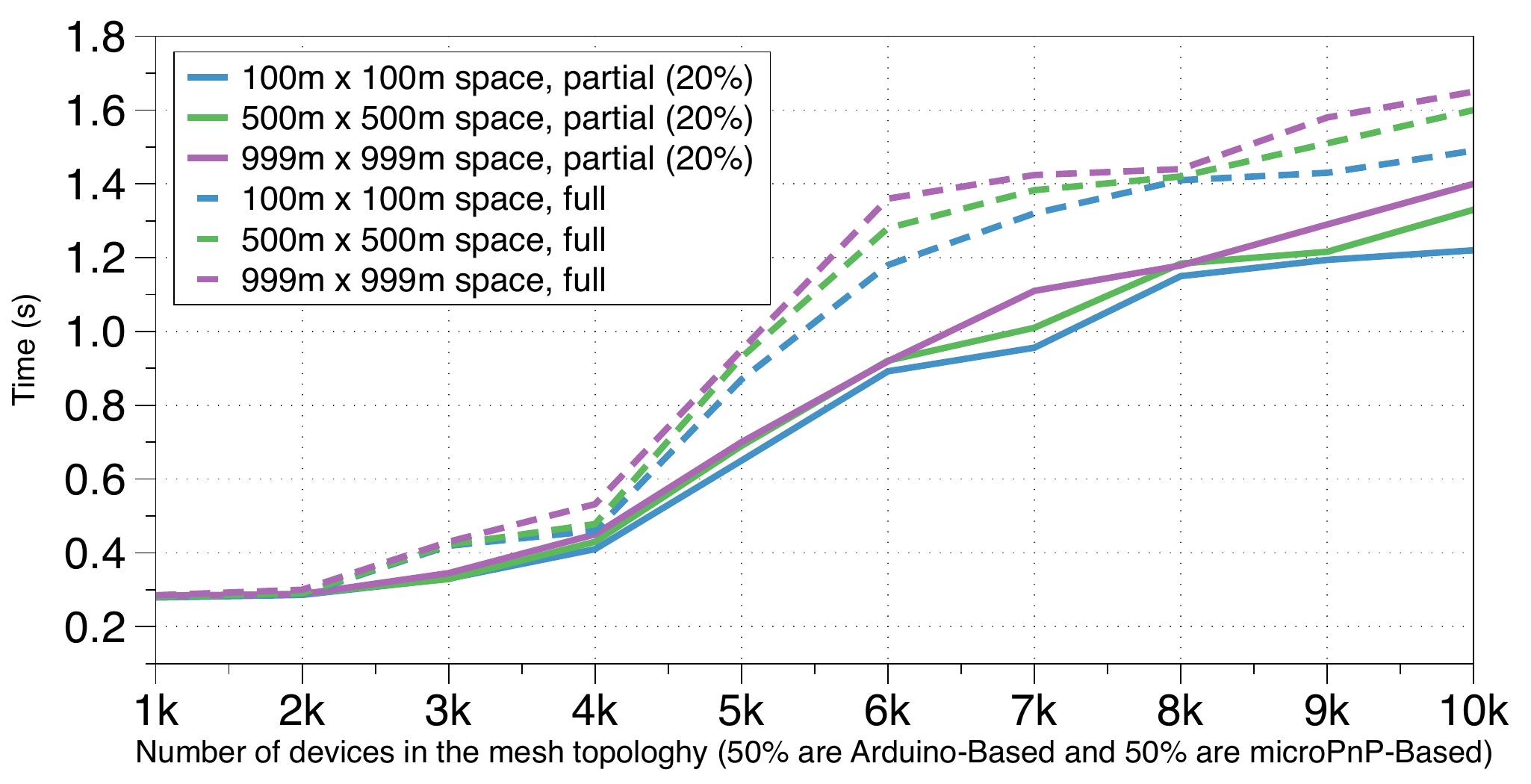}
		\caption{Runtime of static mesh networks}
		\label{fig:mesh}
	\end{subfigure}
	\begin{subfigure}[b]{0.49\textwidth}
		
		\includegraphics[width= \textwidth]{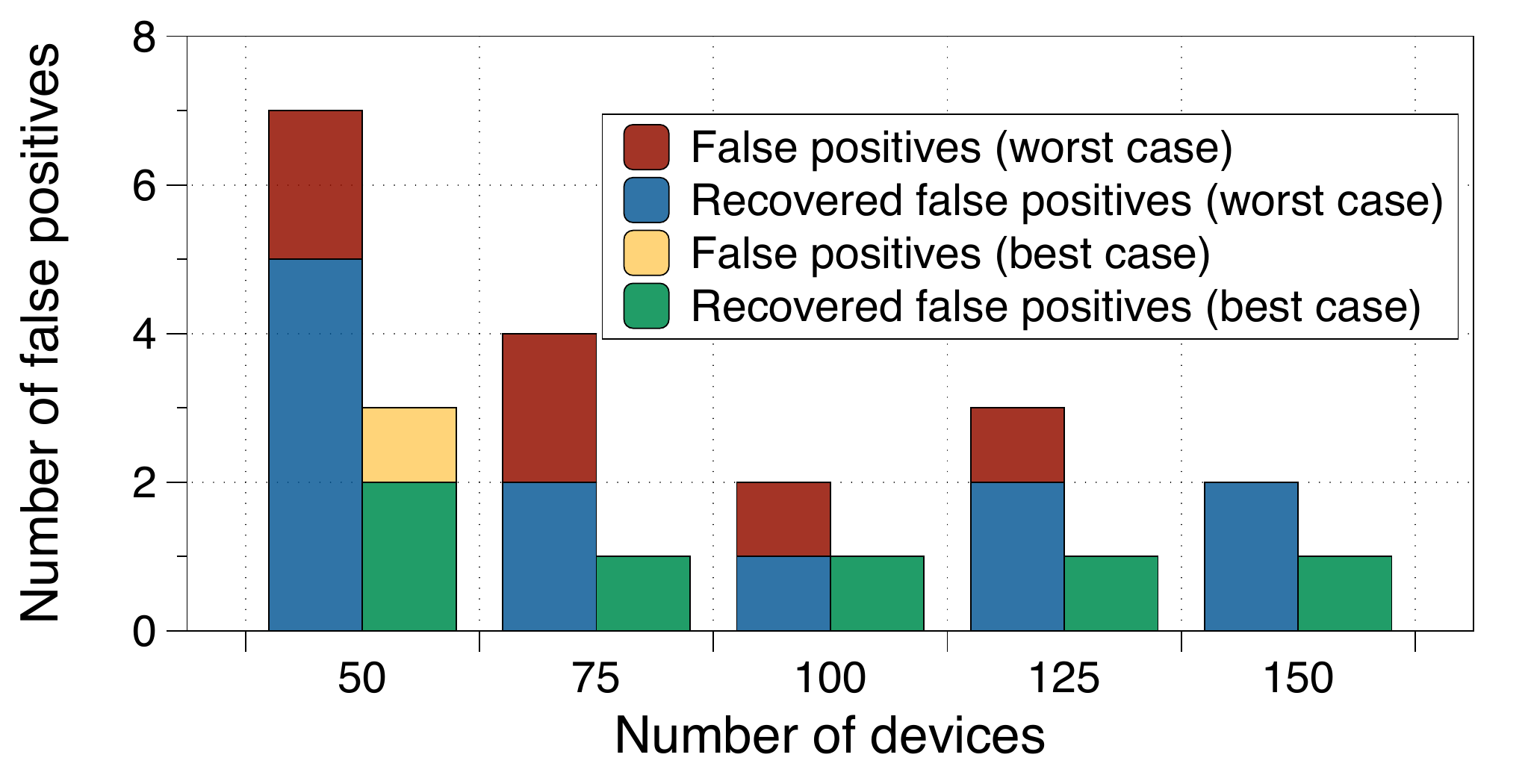}
		\caption{The robustness to false positives} 
		\label{fig:fb}
	\end{subfigure}
	\caption{Evaluation of slimIoT properties}
	\label{figureapp1}
\end{figure*}

\textbf{Memory Footprint.}
According to \cref{init}, each $\rho$ stores (i) its own ID along with the parent and cluster IDs (each 3 bytes), (ii) two private keys ($K_a$ \& $K_t$) and one cluster-wide key ($K_c$) (each 16 bytes), and (iii) a group-wide commitment value ($K_0$), a secret $nonce$, and the true MAC digest of memory ($H_S$) (each 32 bytes). In total, each $\rho$ requires 153 bytes of permanent storage for data in addition to another 64 bytes for storing $H'_S$ and the last authenticated session key. This value is constant and minimal comparing to all other existing approaches. More precisely, assuming that the minimum amount of storage is $c$, each $\rho$ in slimIoT requires only $c$, whereas in SCAPI, each $\rho$ requires $c$ + $l*g$, where $l$ is the length of the channel key and $g$ is the number of neighbours of $\rho$.


\textbf{Scalability.}
 OMNeT++ framework \cite{omnet} is used to simulate large networks with different configurations. The computational and network delays are adjusted according to the experimentally measured values illustrated previously. We set the value of the disclosure delay $d$ to 30 ms.

First, we simulate slimIoT in various static and homogeneous network topologies with a large number of connected IoT devices. 
Figure \ref{fig:trees} shows the total runtime of slimIoT with different configurations for binary, 4-ary, and 8-ary tree topologies consisting of Arduino-based nodes. 
Considering a swarm of 100K devices equally distributed over 8 clusters, attesting the entire swarm against physical attacks in addition to verifying the software integrity of 2 clusters (25\% of the devices) only against remote attacks consumes around 30 seconds in the binary tree, whereas it demands less than 17 seconds in a 4-ary tree and 9 seconds in an 8-ary tree. Performing full attestation of the same network consumes  at most 18 seconds in an 8-ary tree and not more than 50 seconds in a binary tree. 

\textbf{Physical Attack Detection.} Considering the aforementioned scenario in addition to $\upsilon$ of type raspberry pi 3,  identifying physically-compromised devices in  an n-bit vector of length 100K demands a few milliseconds whereas it consumes 6.3 seconds to verify the aggregated attestation reports from 2 clusters (25\%). 
Assuming that $\upsilon$ detected a compromised device in an 8-cluster mesh network, Figure \ref{fig:updates} shows the time consumed to update the group-wide secrets for different network topologies. In a binary tree with 100K devices (the worst case), updating healthy clusters (7 out of 8) consumes at most 0.4 second, whereas it would last up to 2 minutes to update the last healthy device in the compromised cluster.

\textbf{Heterogenous Mesh Networks.} We simulated three static heterogenous mesh networks with different configurations and number of devices. Figure \ref{fig:mesh} shows the runtime overhead incurred by running either partial (20\%) or full attestation in each mesh topology with up to 10K devices. For each mesh topology, we set diverse values for the size of the covered geographical area as well as the communication range of devices (20 m, 40 m, or 60 m). 
%

\textbf{Robustness.}
\label{rob-eva}
Sometimes,  factors such as device mobility and network delay prevent operative devices in the swarm from receiving some updates. This causes false positive cases where a healthy device is mistakenly considered as physically-compromised. 
In slimIoT, If the lost packet carries the secret session key that has to be exposed, the device can still get and verify this key upon receiving the next secret key. For example, assuming that there is a device $D_i$  received  a $Nonce_{update}$ packet, $P1$, in time interval $1$, and another $Attest_{req}$ packet, $P2$, in time interval $2$. So far, $D_i$ cannot authenticate any packet yet. Lets assume that the packet that discloses key $K_1$ is lost. So, $D_i$ still can not authenticate $P1$. In the next disclosure period, $\upsilon$ exposes $K_2$. Upon receiving this key, $D_i$ is able to authenticate both packets, by verifying $K_0$ = $F(F(K_2))$, and learning $K_1$ = $F(K_2)$.
This important feature is not offered by SCAPI\footnote{The authors of SCAPI offer an extension of it to minimize false positive cases. However, this extension is still expensive in terms of communication overhead as it involves changing the leader device with each false positive case and advertising about this change to the entire swarm.} where missing one packet by any $D_i$ renders it directly compromised.  

To evaluate the robustness of slimIoT to false positives in highly dynamic networks, we randomly deployed various small sets of devices in 999 m $\times$ 999 m square area. In each set, we set 20\% of the devices as stationary and let the remaining ones move randomly at speed of 10 m/s. We adjusted the communication range of all devices to 50 m. We then run slimIoT 100 times in each group to record the worst and best case (excluding the optimal case) in terms of number of false positives detected and number of recovered ones. Figure \ref{fig:fb} shows the simulation results. For example, in a network with 50 devices, the worst case shows that there were likely 7 devices to appear as false positives since they were out of range of the connected network and accordingly missed the first session secret key required to authenticate the $Nonce_{update}$ request. Five of which were able to receive the second session secret key and thus authenticate both $Nonce_{update}$ and $Attest_{req}$. Therefore, the worst case ended up with 2 false positives instead of 7. For the same network, the best case shows that 2 out of 3 devices are recovered and so forth 1 false positive occurred. It is clear that the more dense the network, the less false positives likely to appear due to the high probability of recovery.

%% file: sections/conclusion.tex
\section{Conclusion}
\label{sec:conclusions}

This paper introduced slimIoT, an informative lightweight scalable attestation scheme that is compatible with all classes of IoT devices. slimIoT provides strong security guarantees with the presence of a physical attacker. The overhead incurred by slimIoT on the most resource-constrained devices in terms of memory footprint and communication overhead is very reasonable, while simulation results demonstrated that it is scalable to large mesh networks consisting of thousands of  heterogenous IoT devices. Furthermore, it is robust in dynamic networks where false positive cases are unlikely to occur due to the possibility of devices to recover and synchronize with the ongoing attestation depending on the keychain property. 

\section*{Acknowledgment}

This research is supported by the research fund of KU Leuven and IMEC, a research institute founded by the Flemish government.